\documentclass[conference]{IEEEtran}
\usepackage[utf8]{inputenc}
\usepackage{dsfont}

\usepackage{richard_config}
\usepackage{empheq}

\definecolor{myblue}{rgb}{.8, .8, 1}

\ifCLASSINFOpdf
  \usepackage[pdftex]{graphicx}
  \usepackage{epstopdf}
  \graphicspath{{./figures/}}
  \DeclareGraphicsExtensions{.eps}
\else
\fi

\makeatletter
 \def\@textbottom{\vskip \z@ \@plus 11pt}
 \let\@texttop\relax
\makeatother

\author{\IEEEauthorblockN{Abdoulaye Tall\IEEEauthorrefmark{1}, Zwi Altman \IEEEauthorrefmark{1} and Eitan Altman\IEEEauthorrefmark{2}} \\ \IEEEauthorblockA{\IEEEauthorrefmark{1}Orange Labs
38/40 rue du General Leclerc,92794 Issy-les-Moulineaux \\Email: \{abdoulaye.tall,zwi.altman\}@orange.com}\\ \IEEEauthorblockA{\IEEEauthorrefmark{2}INRIA Sophia Antipolis, 06902 Sophia Antipolis, France, Email:eitan.altman@sophia.inria.fr}
}

\title{Self Organizing strategies for enhanced ICIC (eICIC)}
\begin{document}
\maketitle

\begin{abstract}
	Small cells have been identified as an effective solution for coping with the important traffic increase that is expected in the coming years. But this solution is accompanied by additional interference that needs to be mitigated. The enhanced Inter Cell Interference Coordination (eICIC) feature has been introduced to address the interference problem. eICIC involves two parameters which need to be optimized, namely the Cell Range Extension (CRE) of the small cells and the ABS ratio (ABSr) which defines a mute ratio for the macro cell to reduce the interference it produces. In this paper we propose self-optimizing algorithms for the eICIC. The CRE is adjusted by means of load balancing algorithm. The ABSr parameter is optimized by maximizing a proportional fair utility of user throughputs. The convergence of the algorithms is proven using stochastic approximation theorems. Numerical simulations illustrate the important performance gain brought about by the different algorithms.

\begin{IEEEkeywords}
Self-Organizing Networks, enhanced Inter Cell Interference Coordination, eICIC, Load Balancing, Stochastic Approximation, Cell Individual Offset, Almost Blank Sub-Frame (ABS)
\end{IEEEkeywords}
\end{abstract}

\section{Introduction}
	Small cells have been identified as one of the most promising solutions for coping with the expected traffic growth in the coming years. The low transmit power of these nodes makes their footprint very small, limiting the amount of macrocell traffic they can offload. To enhance the offloading capabilities of small cells, the \ac{CRE} mechanism has been introduced in the 3GPP standard allowing a small cell to extend its coverage by increasing the \ac{CIO} parameter. The \ac{CIO} defines the attachment rule of \acp{UE} to the small cell. The macro \ac{BS} which previously served the \acp{UE} at the small cells' extended coverage zone may now strongly interfere them. An interference mitigation mechanism has been proposed to protect these \acp{UE} by muting almost all macro BS transmissions on a certain portion of subframes, thus letting small cells serve their users with less interference. The muted sub frames are denoted as \ac{ABS}. The combination of these two mechanisms has been described in 3GPP \cite[Section 16.1.5]{3gpp.36.300-R10} under the name of \textit{\ac{eICIC}}.
	
Two different parameters are involved in \ac{eICIC}: the \ac{CIO} which in our scenario is applied only to the small \acp{BS} to offload macro BSs, and the \ac{ABSr}, namely the percentage of the muted subframes at the macro BSs. Although the principle of eICIC has been described in 3GPP, its actual implementation is not clearly specified, and in particular, the way to set these eICIC parameters which may be critical for the cell performance.

Previous contributions on eICIC focus on performance gain brought about by this feature in a static heterogeneous network \cite{wang2012performance,pedersen2012eicic,shirakabe2011performance}. In \cite{pang2012optimized} and \cite{bedekar2013optimal}, the \ac{eICIC} parameters are optimized for a static scenario. The \ac{ABSr} optimization problem has been treated for a dynamic environment in \cite{vasudevandynamic} assuming fixed \acp{CIO}. The problem of optimizing both \ac{CIO} and \ac{ABSr} has been studied in \cite{deb2013algorithms,yasireICIC} using a centralized approach. The solution is computationally demanding and can be implemented as a management plane solution.
	
The purpose of this paper is to propose efficient and distributed \ac{SON} algorithms for optimizing both the \acp{CIO} and the \acp{ABSr} parameters using stochastic approximation techniques \cite{Kushner,Borkar}.
We use a distributed load balancing objective for the small cell \ac{CIO} optimization \cite{CombesInfocom2012}. 
Then we propose a distributed solution for optimizing the \acp{ABSr} at the small \acp{BS} which in turn request these \acp{ABSr} from their interfering macro \acp{BS}. \ac{ABSr} optimization algorithms are proposed using objective functions based on a \ac{PF} utility of users' throughputs. Different strategies for scheduling small \acp{BS}' users during the muted subframes of the macro \ac{BS} are studied.

	The contributions of the paper are the following:
	\begin{itemize}
	\item A load balancing \ac{SON} algorithm for \ac{CRE} optimization using results from \cite{CombesInfocom2012}. 
	\item A distributed scheme for optimizing the \ac{ABSr} of the macro \acp{BS}.
	\item Two \ac{SON} algorithms based on \ac{SA} for optimizing \ac{ABSr} using \ac{PF} utility functions.
	\item Performance evaluation of the SON algorithms in a dynamic environment taking into account flow level dynamics of elastic traffic.
	\end{itemize}

The paper is organized as follows. Section II describes the \ac{eICIC} mechanism, and analyzes the necessary conditions for achieving gain using \ac{eICIC}. Two implementation alternatives of \ac{eICIC} are also proposed. Section \ref{sec:son_algorithms} presents the main \ac{SA} algorithms that allows to implement an adaptive optimal \ac{eICIC}. In section \ref{sec:perf_results}, the performance results of those algorithms in a heterogeneous network with realistic traffic are illustrated. Section \ref{sec:conclusion} concludes the paper.

\section{Problem Statement}

\subsection{eICIC mechanisms}
	The deployment of small cells presents technological challenges. The low transmit power of small \acp{BS} and antenna height limits their coverage area and hence their offloading capabilities. To order to overcome this problem, \ac{CRE} can be used.

\subsubsection{\ac{CRE}}
	\ac{CRE} is performed by increasing the \ac{CIO} of the small \ac{BS}. \ac{UE} attachment to a \ac{BS} is determined by comparing the received pilot powers from all surrounding \acp{BS} plus a certain offset which is denoted as \acf{CIO}. The attachment rule for \ac{UE} $u$ can be formulated as follows
\begin{equation}
s^* = \text{argmax}_s \text{CIO}_s h_s^u P_s
\end{equation}
where $s^*$ is the chosen serving cell, $\text{CIO}_s$ - the \ac{CIO} of cell $s$, $P_s$ - its pilot power and $h_s^u$ - the pathloss from \ac{BS} $s$ to \ac{UE} $u$.

\begin{figure}[!ht]
\centering
\includegraphics[width=3in]{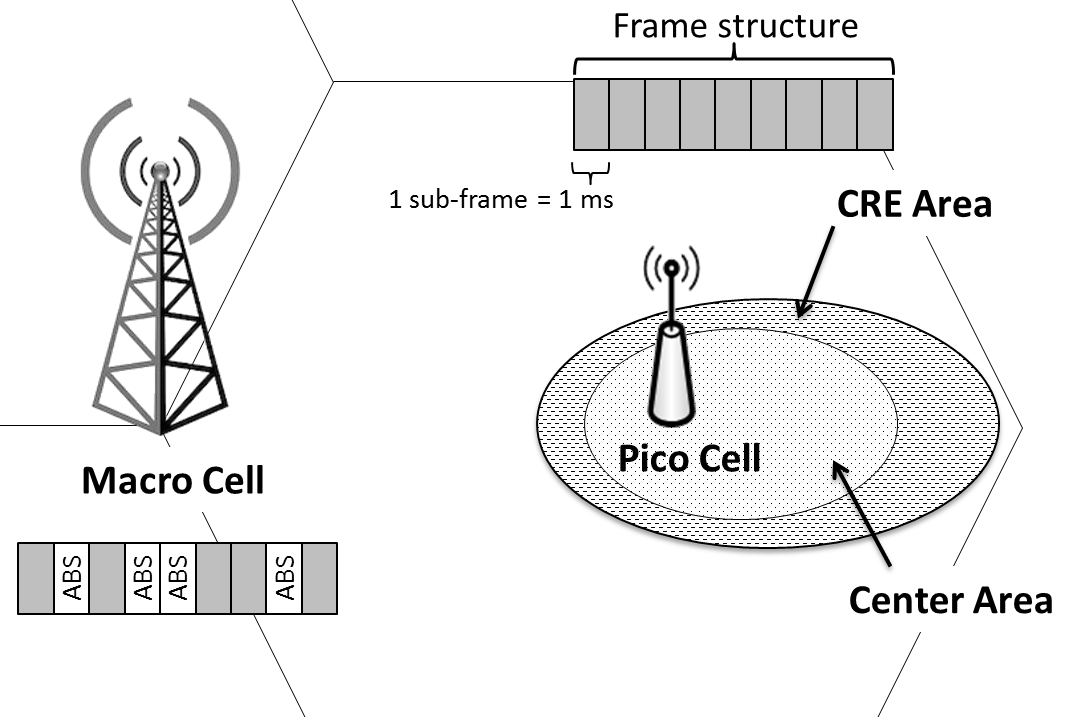}
\caption{Illustration of Cell Range Extension and Almost Blank Sub-Frames in a HetNet}
\label{fig:cre_abs}
\end{figure}

	By setting a nonnegative $\text{CIO}_{dB}$ (\ac{CIO} in dB) at the small \ac{BS} we force the macro users to attach to a small cell which is not their best serving cell as shown in Figure \ref{fig:cre_abs}. We say that these users are in the \ac{CRE} area. Hence the \ac{CRE} allows to increase the offloading of the macro cell. However, the \ac{CRE} users now experience more interference from the macro \ac{BS} which is their best serving cell. The macro cell interference reduces the \ac{SINR} at the \ac{CRE} area thus limiting the extent of small cell offloading. The problem of interference mitigation for \ac{CRE} users has been addressed by 3GPP and the \acp{ABS} mechanism has been proposed.  

\subsubsection{\acp{ABS}}
	The \ac{ABS} mitigation method consists in a time-domain interference avoidance. The goal is to reduce the interference from an aggressor cell (the cell causing the interference, which in our case is the offloaded macro \ac{BS}) by almost blanking out some of its sub-frames as shown in Figure \ref{fig:cre_abs}.

	During the \acp{ABS}, the aggressor cell mutes all of its traffic channels, leaving only certain control channels which are transmitted with reduced power. The \acp{ABS} allows victim cells (namely the interfered small \acp{BS}) to serve their users with almost no interference from the aggressor cell. The residual interference caused by the remaining control signals limits the gain from the \acp{ABS}, and can be further mitigated using additional interference cancelation techniques introduced in 3GPP release 11 \cite[Section 16.1.5]{3gpp.36.300-R11}.
	
	It is advisable to schedule cell edge users during \acp{ABS} because they are the most affected by interference. However, all the small cell users can benefit from \acp{ABS} since their strongest interferer is generally a macro \ac{BS}. The use of \acp{ABS} decreases the available resources of the macro \ac{BS}. Hence a trade-off should be found between the capacity gain of the small cells brought about the \acp{ABS} and the corresponding capacity losses of the macro \ac{BS}.
	
	One may choose to schedule small cell \acp{UE} in the \ac{CRE} area exclusively during the \acp{ABS}, and schedule the other small cell \acp{UE} only during non-\acp{ABS}. Alternatively, one may schedule all small cells' \acp{UE} during both \acp{ABS} and non-\acp{ABS}. We now discuss these two possible implementations.

\subsection{Implementation alternatives}
\subsubsection{Protection of offloaded users} \label{Protectedcase}
	The first implementation for \ac{eICIC} aims at protecting only the offloaded users at the \ac{CRE} area of the small cells.
Users of the small cells are divided into two groups: \ac{CRE} users who are attached to the small cell due to the \ac{CIO}, and center cell users who are attached to the small cell even when the \ac{CIO} is set to $0$dB. The \ac{CRE} users will be served by the small cell during the macro cell \acp{ABS}. Note that this means a strictly positive \ac{CIO} must be accompanied by a strictly positive \ac{ABSr}, otherwise the users in the \ac{CRE} area will never be served. We say that the two parameters (\ac{CIO} and \ac{ABSr}) are coupled.
		
	Consider a \ac{CRE} user $u$ and determine the \ac{SINR} gain when he is offloaded from the macro cell to the small cell and is scheduled during \acp{ABS}. Denote by the subscript $m$ the macro \ac{BS} and by $p$ - the small \ac{BS}. When attached to the macro \ac{BS} $m$, user $u$ has a \ac{SINR} equals to
	\begin{equation}
	S_{u,m} = \frac{h_m(u) P_m}{h_p(u) P_p + C_0(u)}
	\end{equation}
	where $h_m(u)$ and $h_P(u)$ are the pathlosses from the macro cell $m$ to user $u$ and from the small cell $p$ to user $u$, respectively. $P_m$ and $P_p$ are the macro cell and small cell transmit powers, respectively. $C_0(u)$ is the total interference generated by the other nodes in the network (other macro \acp{BS} and small \acp{BS}) plus the thermal noise at the receiver of user $u$. If user $u$ is offloaded to the small \ac{BS} and is served only during the \acp{ABS} of macro cell $m$, then its \ac{SINR} becomes
	\begin{equation}
	S_{u,m} = \frac{h_p(u) P_p(u)}{C_0(u)}
	\end{equation}
	The \ac{SINR} gain for this user can then be written as
	\begin{equation}
	\begin{split} \label{eq:SINR gain}
	SG_u & = \frac{h_p(u) P_p (h_p(u) P_p + C_0(u))}{h_m(u) P_m C_0(u)} \\
	     & = \frac{h_p(u) P_p}{h_m(u) P_m} + \frac{(h_p(u) P_p)^2}{h_m(u) P_m C_0(u)}
	\end{split}
	\end{equation}
	From equation \eqref{eq:SINR gain} we can deduce the following simple condition on the received signals from the different \acp{BS} at every point in the considered area for obtaining an offloading gain using \acp{ABS}
	\begin{equation} \label{eq:interf_max}
		\frac{(h_p(u) P_p)^2}{h_m(u) P_m - h_p(u) P_p} > C_0(u)
	\end{equation}
	If the condition \eqref{eq:interf_max} is not satisfied, offloading will result in performance degradation. Furthermore, there is a maximum \ac{CIO} above which there is no gain for certain small cell users since as one gets further away from the small \ac{BS}, $C_0(u)$ increases compared to the received signal strength from the small \ac{BS}.
	
	Condition \eqref{eq:interf_max} considers that only one macro cell implements \acp{ABS}. If $M$ macro cells implement \acp{ABS}, their interference will be removed from $C_0(u)$ so that condition \eqref{eq:interf_max} becomes
	\begin{equation} \label{eq:gen_interf_max}
		\frac{(h_p(u) P_p)^2}{ \sum_{m=1}^M h_m(u) P_m - h_p(u) P_p} > C_0(u)
	\end{equation}
	which is verified for a larger number of \ac{CRE} users.
	
	We now give a simple example in which we can find at which extent a \ac{CRE} can be performed, i.e. where offloading provides \ac{SINR} gains. We consider a trisector macro cell site with one small \ac{BS} in the coverage area of one of the macro sectors. To take into account neighbours' interference, we add a tier of trisector macro sites to this cluster.

	We focus on the \ac{SINR} gains for users in the \ac{CRE} area. By varying the number of macro \acp{BS} applying \acp{ABS} for the small cell users, we can see the evolution of the maximum \ac{CIO} above which there is a \ac{SINR} degradation at the edge of the small cells. The macros considered for muting are the most interfering ones.

\begin{figure}[!ht]
\centering
\includegraphics[width=3in]{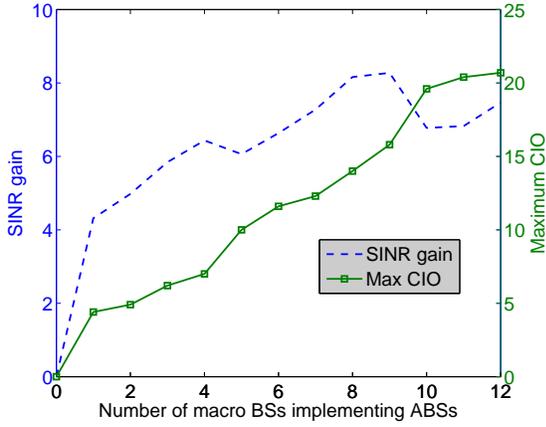}
\caption{Maximum CIO for SINR improvement and SINR gain as a function of the number of macro \acp{BS} applying \acp{ABS}}
\label{fig:max_cio_evolution}
\end{figure}

	Figure \ref{fig:max_cio_evolution} presents the maximum \ac{CIO} as a function of the number of macro \acp{BS} providing \acp{ABS} (red line with squares). It also shows the mean \ac{SINR} gain obtained by the \ac{CRE} users when setting the \ac{CIO} to its maximum value (blue dashed line). The Figure illustrates condition \eqref{eq:gen_interf_max}, namely that muting more macro \acp{BS} allows us to further increase the \ac{CIO}. We can also see that the \ac{SINR} gain obtained by muting more and more macro \acp{BS} eventually saturates. Hence an optimal number of macro \acp{BS} applying \ac{ABS} can be found while keeping complexity low.	
	
	It is noted that if the small cells' load is low, an offloading is still possible without applying the \ac{ABS} which reduces the resources available for the macro cell. We present in the next section an implementation that allows to achieve this case.
	
\subsubsection{Protection of all small cell users}\label{NonProtectedCase}
	Instead of allowing only the \ac{CRE} users to profit from \ac{ABS}, in this implementation all the small cells' users can benefit from the \ac{ABS}, namely they can be scheduled during \ac{ABS} or during normal sub-frames. In this case, offloading by increasing the \acp{CIO} of the small cells without applying \acp{ABS} is also allowed. We can say that the two parameters (\ac{CIO} and \ac{ABSr}) are decoupled. However if needed, the macro cell can provide \ac{ABS} in order to enhance the offloading capability of the small cells. The \ac{SINR} during \acp{ABS} period is the same as that in the previous implementation. Only now, a small cell user will benefit part of the time from \acp{ABS} and the rest of the time will experience normal \ac{SINR}. The mean throughput of a macro user can then be written as
	\begin{equation} \label{eq:mac_ue_thrpt_abs}
		\bar{T}_{u,m} = (1-\theta) \bar{R}_{u,m}
	\end{equation}
where $\theta$ is the \ac{ABSr} of macro $m$ and $\bar{R}_{u,m}$ - the mean data rate of user $u$ when it is served by macro cell $m$.
The mean throughput of a small cell user is
	\begin{equation}\label{eq:small_ue_thrpt_abs}
		\bar{T}_{u,p} = (1-\theta) \bar{R}_{u,p}^{\text{no ABS}} + \theta \bar{R}_{u,p}^{\text{ABS}}
	\end{equation}
	where $\theta$ is the \ac{ABSr} available to small cell $p$, $\bar{R}_{u,p}^{\text{no ABS}}$ and $\bar{R}_{u,p}^{\text{ABS}}$ are the average data rates of user $u$ over time when served by small cell $p$ outside and during the \acp{ABS} respectively.
	
	An efficient setting of the \ac{ABSr} is needed to optimize a function of the users' throughputs. A condition for achieving an offloading gain using \acp{ABS} could be the increase of the total throughput of the considered cluster (macro cells and small cells)
	\begin{equation}\label{eq:ABS loose condition}
	\begin{split}
		& \sum_{m \in \mathcal{M}} \sum_{u \in m} (1-\theta) \bar{R}_{u,m} \\
		& + \sum_{p \in \mathcal{P}} \sum_{u \in p} (1-\theta) \bar{R}_{u,p}^{\text{no ABS}} + \theta \bar{R}_{u,p}^{\text{ABS}} \\
		& \geq \sum_{m \in \mathcal{M}} \sum_{u \in m}\bar{R}_{u,m} + \sum_{p \in \mathcal{P}} \sum_{u \in p} \bar{R}_{u,p}^{\text{no ABS}}
	\end{split}	
	\end{equation}
	where $\mathcal{M}$ and $\mathcal{P}$ are the set of all macro and small \acp{BS} involved in the mechanisms (\ac{CRE} and \ac{ABS}).
	
	Condition \eqref{eq:ABS loose condition} may be too restrictive as we may want to increase the \ac{CET} at the expense of a decrease in total throughput. In the next section, we propose self-optimization algorithms based on the two implementations described in this section, using a PF utility of \ac{UE} throughputs as objective function.
	
\section{SON Algorithms} \label{sec:son_algorithms}

\subsection{Load Balancing SON}
\ac{LB} \ac{SON} adjusts the small \acp{BS} \acp{CIO} in order to balance the loads between a macro \ac{BS} and the small \acp{BS}. The idea is to increase the small cell whenever its load is lower than that of the macro \ac{BS} in the coverage area of which it is located (and vice versa). If we consider traffic offloading from macro $m$ to a small cell $s$, the \acp{ODE} defining the \ac{LB} \ac{SON} mechanism at the small \ac{BS} is defined by
\begin{equation}\label{LBode}
\dot{CIO}_{{\text{dB}}_s} = \rho_m({\text{dB}}_s) - \rho_s({\text{dB}}_s)
\end{equation}
where ${CIO_{\text{dB}}}_s$ is the \ac{CIO} of small \ac{BS} $s$ in decibels, $\rho_m$ - the load of the macro \ac{BS} $m$, and $\rho_s$ - the load of the small \ac{BS} $s$.
	The \ac{SA} update equation defining the \ac{SON} algorithm for \eqref{LBode} reads
\begin{equation} \label{eq:LB_SA}
	CIO_s(k+1) = CIO_s(k) + \epsilon (\hat{\rho}_m(k)-\hat{\rho}_s(k))
\end{equation}
where $\hat{\rho}_m$ and $\hat{\rho}_s$ are the estimators of $\rho_m$ and $\rho_s$ obtained by averaging the resource utilization of the respective \ac{BS} over a certain time period. This SON converges to a set on which all loads are equal as shown in \cite[Theorem 4]{CombesInfocom2012}.
	
	In practice, a projected \ac{SA} algorithm will be used instead of \eqref{eq:LB_SA} because the \acp{CIO} are restricted to a maximum value. The restriction is mainly due to the fact that offloaded users suffer from residual interference caused by remaining control channels during \acp{ABS}. However, the convergence of the \ac{SA} remains valid even in this case (see \cite[§5.4]{Borkar}).

\subsection{ABS ratio optimization}
	We choose to implement the \ac{ABSrO} algorithm at the small \ac{BS} which then requests appropriate \ac{ABSr} from its interfering macro \acp{BS}. The macro \acp{BS} receives \ac{ABSr} requests from the small cells it interferes and then applies the maximum \ac{ABSr} among these requests. The \ac{ABSr} optimization algorithm should then take into account load or traffic conditions (i.e. number of users present in the cell) of all the macro cells from which the considered small cell will be requesting \acp{ABS}.
	
	The cluster of \acp{BS} considered for a single \ac{ABSrO} algorithm comprises a small cell $p$ on which the algorithm is implemented, and the most interfering $M$ macro cells with small cell $p$. Typically $M=1$ is sufficient if the small cell is in the center of the macro cell, namely by periodically muting only one macro cell, we increase significantly the \ac{SINR} of the small cell users. When the small cell is located at the cell edge, choosing $M=3$ provides better results.

\subsubsection{Only \ac{CRE} users are protected} \label{protected_case_son}
	Using \ac{CRE}, the small \acp{BS} are able to offload traffic of the macro \ac{BS}. The offloaded users at the small cell edge are highly interfered by the macro that previously served them. \acp{ABS} are used to mitigate the interference enabling the small cells to offload even more the macro cell traffic.	
	
	The objective function considered in this implementation is the \ac{PF} defined as follows
	\begin{equation}
	\begin{split}
		U_{PF1}(\theta) = & \sum_{m=1}^M \sum_{u \in m} \log( (1-\theta) \bar{R}_{u,m}) \\
		                + & \sum_{u \in \text{center of } p} \log((1-\theta) \bar{R}_{u,p}^{\text{no ABS}}) \\
		                + & \sum_{u \in \text{CRE of } p} \log( \theta \bar{R}_{u,p}^{\text{ABS}})
		\end{split}
	\end{equation}
	where we considered the $M$ most interfering macro \acp{BS} users throughputs, the small cell center user throughputs and the small cell \ac{CRE} area users throughputs.
	The \ac{PF} utility enables us to maximize the users throughput and to enforce fairness among them. The \ac{SON} algorithm applied to optimize this PF utility is given in the following theorem
	
\begin{theorem} \label{th:absr_pf1}
	Given some positive step sizes $\epsilon_k$ non-summable ($\sum_{k=0}^\infty \epsilon_k = \infty$) but square-summable ($\sum_{k=0}^\infty \epsilon_k^2 < \infty$) and the update equation
	\begin{equation} \label{eq:ABSrO_PF1_SA}
		\theta_{k+1} = \theta_k + \epsilon_k \left(\frac{N_{p,\textit{CRE}}}{\theta} - \frac{N_{p,\textit{CEN}} + \sum_{m=1}^M N_m}{1-\theta} \right)
	\end{equation}
	where $N_m$, $N_{p,\textit{CEN}}$ and $N_{p,\textit{CRE}}$ are the mean numbers of active users in cell $m$, the center of cell $p$ and the \ac{CRE} area of cell $p$, respectively,
	
then $\theta_k$ converges to a set on which $U_{PF1}()$ is maximal.
\end{theorem}
\begin{proof}
	See Appendix \ref{app:absr_pf1}.
\end{proof}
	
	Note that the optimal $\theta$ can be directly derived in this case using \eqref{eq:pf1_der} \cite[Eq. 8]{vasudevandynamic}, and is equal to

	\begin{equation}
	\theta^* = \frac{N_{p,\textit{CRE}}}{N_{p,\textit{CRE}}+N_{p,\textit{CEN}} + \sum_{m=1}^M N_m}
	\end{equation}
	
	The reason we use a \ac{SA} algorithm instead of setting optimal $\theta$ is to keep a certain stability in the parameter configuration which can be quite critical in real networks. Figure \ref{fig:theta_saVSopt} illustrates this statement. Setting optimal \ac{ABSr} at each event in the network (arrival or departure of a user) yields an extremely fluctuating parameter whereas the \ac{SA} approach allows us to freeze the parameter at convergence when the traffic is stationary giving the same performance results.

\begin{figure}[!ht]
\centering
\includegraphics[width=3in]{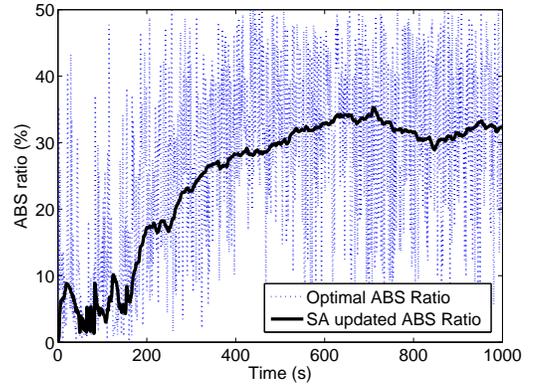}
\caption{Time evolution of ABS ratio using stochastic approximation (solid line) and optimal solution (dashed line)}
\label{fig:theta_saVSopt}
\end{figure}

Performance results obtained using equation \eqref{eq:ABSrO_PF1_SA} are presented in Section \ref{sec:perf_results}. We now proceed with the second implementation for \ac{eICIC}.

\subsubsection{All small cell users benefit from \acp{ABS}}
	In this case, the user throughputs are modified according to equations \eqref{eq:mac_ue_thrpt_abs} and \eqref{eq:small_ue_thrpt_abs}.
	
The \ac{PF} utility of the user throughputs is redefined as
	\begin{equation} \label{eq:U_PF2_exact}
	\begin{split}
		U_{\text{PF2\_exact}}&(\theta) = \sum_{m=1}^M \sum_{u \in m} \log((1- \theta) \bar{R}_{u,m}) \\
		                & + \sum_{u \in p} \log((1-\theta) \bar{R}_{u,p}^{\text{no ABS}} + \theta \bar{R}_{u,p}^{\text{ABS}})
	\end{split}	
	\end{equation}
 since we make no distinction between \ac{CRE} users of the small cell and its normal users. The \ac{SON} algorithm for optimizing the utility function \eqref{eq:U_PF2_exact} is presented in the following theorem.
	
\begin{theorem} \label{th:absr_pf2_exact}
Given some positive step sizes $\epsilon_k$ non-summable ($\sum_{k=0}^\infty \epsilon_k = \infty$) but square-summable ($\sum_{k=0}^\infty \epsilon_k^2 < \infty$) and the update equation
	\begin{equation}  \label{eq:ABSrO_PF2_exact_SA}
		\theta_{k+1} = \theta_k + \epsilon_k \left(\sum_{u \in p} \frac{1}{\theta_k + \frac{\bar{R}_{u,p}^{\text{no ABS}}}{\bar{R}_{u,p}^{\text{ABS}} - \bar{R}_{u,p}^{\text{no ABS}}}} - \frac{\sum_{m=1}^M N_m}{1-\theta_k} \right)
	\end{equation}
	where $N_m$ is the mean number of active users in cell $m$,
	
	then $\theta_k$ converges to a set on which $U_{\text{PF2\_exact}}()$ is maximal.
\end{theorem}
\begin{proof}
	See Appendix \ref{app:absr_pf2_exact}.
\end{proof}
	
	The update equation \eqref{eq:ABSrO_PF2_exact_SA}, 
requires the knowledge of the average data rates of all pico users, rendering practical implementation of the algorithm more complex. To simplify \eqref{eq:ABSrO_PF2_exact_SA}, we choose to maximize a lower bound of \eqref{eq:U_PF2_exact} instead. The new objective function is written as
\begin{equation} \label{eq:U_PF2}
	\begin{split}
		U_{PF2}(\theta) = & \sum_{m=1}^M \sum_{u \in m} \log((1- \theta) \bar{R}_{u,m}) \\
		                + & \sum_{u \in p} \frac{1}{2}\log(2 (1-\theta) \bar{R}_{u,p}^{\text{no ABS}}) \\
		                + & \sum_{u \in p} \frac{1}{2} \log( 2 \theta \bar{R}_{u,p}^{\text{ABS}})
	\end{split}	
\end{equation}

\begin{lemma}
	Let us consider $M$ macro cells and one small cell indexed $p$.

	Then $\forall \theta \in ]0, 1[$ we have
\begin{equation}
U_{PF2}(\theta) \leq U_{\text{PF2\_exact}}(\theta)
\end{equation}
\end{lemma}
\begin{proof}
	For the proof it suffices to show that for a small cell user $u$,
	\begin{equation}
	\begin{split}	
	\log(2 (1-\theta) \bar{R}_{u,p}^{\text{no ABS}}) + \log( 2 \theta \bar{R}_{u,p}^{\text{ABS}}) \\
	\leq 2 \log((1-\theta) \bar{R}_{u,p}^{\text{no ABS}} + \theta \bar{R}_{u,p}^{\text{ABS}})
	\end{split}
	\end{equation}
	For ease of notation, denote $a = (1-\theta) \bar{R}_{u,p}^{\text{no ABS}}$ and $b = \theta \bar{R}_{u,p}^{\text{ABS}}$. We want to show that $\log 2 a + \log 2 b \leq 2 \log(a+b)$. Using Jensen's inequality \cite{jensen1906fonctions} for the function $-\log$ which is convex we have
	\begin{equation} \label{eq:jensen_simple}
	- \log \left( \frac{1}{2} (2a) + \frac{1}{2} (2b) \right) \leq - \frac{1}{2} \log(2a) - \frac{1}{2} \log(2b)
	\end{equation}
	By taking the negative of \eqref{eq:jensen_simple}, we obtain the desired result.
\end{proof}
	The \ac{SON} algorithm optimizing \eqref{eq:U_PF2} is presented in the following theorem.
	
\begin{theorem} \label{th:absr_pf2}
	Given some positive step sizes $\epsilon_k$ non-summable ($\sum_{k=0}^\infty \epsilon_k = \infty$) but square-summable ($\sum_{k=0}^\infty \epsilon_k^2 < \infty$) and the update equation
	\begin{equation}  \label{eq:ABSrO_PF2_SA}
		\theta_{k+1} = \theta_k + \epsilon_k \left(\frac{N_p}{2 \theta_k} - \frac{\frac{N_p}{2} + \sum_{m=1}^M N_m}{1-\theta_k} \right)
	\end{equation}
	where $N_m$ and $N_p$ are the mean numbers of active users in cells $m$ and $p$ respectively,
	
	then $\theta_k$ converges to a set on which $U_{PF2}()$ is maximal.
\end{theorem}
\begin{proof}
	See Appendix \ref{app:absr_pf2}.
\end{proof}
	
	The reasons for implementing a \ac{SA} algorithm instead of setting the optimal \ac{ABS} ratio (which can be easily derived) are the same as those stated in the previous section.
	
	The \ac{SON} algorithm \eqref{eq:ABSrO_PF2_SA} presents certain advantages over the one in \eqref{eq:ABSrO_PF1_SA}. The first one is that we do not need to keep two counters for the numbers of active users of the small \ac{BS}, but only one for the total number of users. The second one is that if the small \ac{BS} has a very low load compared to those of its surrounding macro \acp{BS}, the \ac{ABSr} provided by those macro \acp{BS} can be also low, thus preserving the resources of the macro \acp{BS}. Furthermore, \eqref{eq:ABSrO_PF2_SA} is completely decoupled from the load balancing \ac{SON}, whereas in \eqref{eq:ABSrO_PF1_SA}, a positive \ac{CIO} requires a corresponding \ac{ABSr}. In the next section, we compare the performance results of the different \ac{SON} algorithms.

\section{Simulation Results} \label{sec:perf_results}
\subsection{Simulation scenario}
	Consider a trisector \ac{BS} surrounded by 6 interfering macro sites. In each macro sector 4 small \acp{BS} are deployed close to the cell edge (see Figure \ref{fig:net_lay}). We consider elastic traffic where users arrive in the network according to a Poisson process, download a file and leave the network as soon as their download is complete. The considered area $A$ is the initial area covered by the macro - and the small \acp{BS}. Two layers of traffic are superposed: the first one has a uniform arrival rate of $\lambda$ users/s all over $A$, and the second - a uniform arrival rate of $\lambda_h$ users/s in the initial area covered by the small cells (with all \acp{CIO} set to 0dB). This is close to a realistic scenario where small cells are deployed in hotspot areas.
	
	Following the \ac{SINR} gain results presented in Figure \ref{fig:max_cio_evolution}, we choose three macro \acp{BS} for \ac{ABSrO} ($M = 3$). Each small \ac{BS} requests \acp{ABS} from its three most interfering macro \acp{BS} and take their load conditions (number of users to serve) into account in the \ac{ABSrO} algorithms. To avoid truncation effects of the computational area in the simulations, we suppose that the surrounding macro \acp{BS} serve 5 active users on the average.

\begin{figure}[!ht]
\centering
\includegraphics[width=3in]{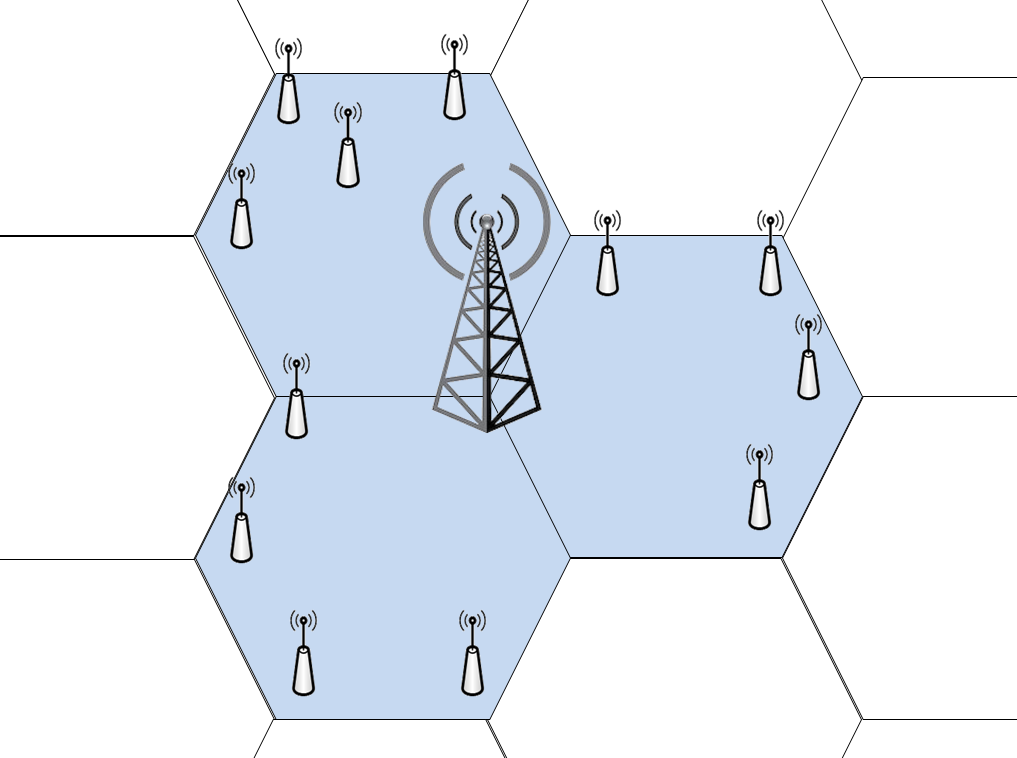}
\caption{Network layout scenario}
\label{fig:net_lay}
\end{figure}

	We use the propagation models for macro - and small \acp{BS} (following \cite[Page 61]{3gpp.36.814-R9}) presented in Table \ref{tab:params} which also summarizes all the simulation parameters.
	
\begin{table}[!t]
\renewcommand{\arraystretch}{1.3}
\caption{Network and Traffic characteristics}
\label{tab:params}
\centering
\begin{tabular}{|c|c|}
\hline
\multicolumn{2}{|c|}{Network parameters} \\
\hline
Number of macro BSs & 3 \\
\hline
Number of small BSs & 12 \\
\hline
Number of interfering macros & 6 $\times$ 3 sectors \\
\hline
Macro Cell layout & hexagonal trisector \\
\hline
Small Cell layout & hexagonal omni \\
\hline
Intersite distance & 500 m \\
\hline
Bandwidth & 10MHz \\
\hline
\multicolumn{2}{|c|}{Channel characteristics} \\
\hline
Thermal noise & -174 dBm/Hz \\
\hline
Macro Path loss (d in km) & 128.1 + 37.6 $\log_{10}(d)$ dB \\
\hline
small cell Path loss (d in km) & 140.7 + 36.7 $\log_{10}(d)$ dB \\
\hline
\multicolumn{2}{|c|}{Traffic characteristics} \\
\hline
Traffic spatial distribution & uniform \\
\hline
$\lambda$ & $\text{14 users/s/km}^2$ \\
\hline
$\lambda_h$ & $\text{6 users/s/km}^2$ \\
\hline
Service type & FTP \\
\hline
Average file size & 10 Mbits \\
\hline
\end{tabular}
\end{table}
	
\subsection{Performance Evaluation}
	We first compare the performance using the exact PF utility function \eqref{eq:ABSrO_PF2_exact_SA} and the approximate PF utility \eqref{eq:ABSrO_PF2_SA} for the second implementation (namely all small cell users can benefit from \acp{ABS}). The results are presented in Figure \ref{fig:pf2_cdf} in terms of the Cumulative Distribution Function (CDF) of user throughputs. We can see that the exact utility (dashed blue line) gives slightly better \acp{MUT} than the approximated PF utility (red stars curve), although the difference is not significant. We therefore consider in the following only the approximated algorithm \eqref{eq:ABSrO_PF2_SA}, motivated by its much lower implementation complexity.
	
\begin{figure}[!ht]
\centering
\includegraphics[width=3in]{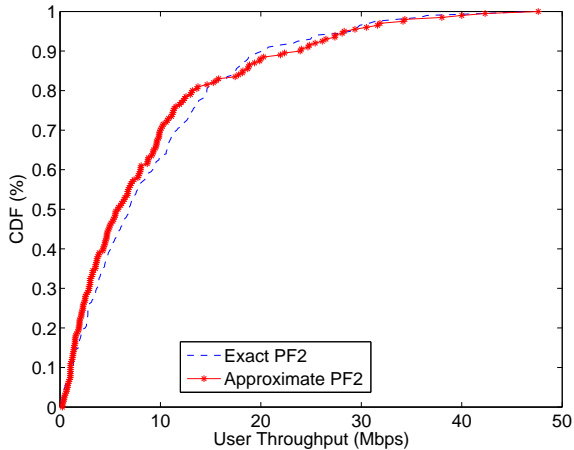}
\caption{Comparison of user throughput CDF using the exact PF2 \eqref{eq:ABSrO_PF2_SA} - and the approximated PF2 \eqref{eq:ABSrO_PF2_exact_SA} utility functions}
\label{fig:pf2_cdf}
\end{figure}

	We next compare the performance results for the following four cases:
	\begin{enumerate}
	\item Case 1: The initial network with no \ac{SON} enabled, denoted hereafter as 'NoSON'. This case serves as a reference.
	\item Case 2: The load balancing \ac{SON} only, using algorithm \eqref{eq:LB_SA}, and denoted as 'LBonly'.
	\item Case 3: The load balancing and \ac{ABSrO} algorithms are enabled under implementation 1 of \ac{eICIC} with \ac{ABSrO} algorithm \eqref{eq:ABSrO_PF1_SA} using the \ac{PF} utility function. This case is denoted as 'PF1'.
	\item Case 4: The load balancing and \ac{ABSrO} algorithms are enabled under implementation 2 of \ac{eICIC} with \ac{ABSrO} algorithm \eqref{eq:ABSrO_PF2_SA} using the approximated PF utility function. This case will be denoted as 'PF2approx'.
	\end{enumerate}
	
	The chosen performance indicators are the global \ac{CET} (Figure \ref{fig:cet}), the global \ac{MUT} (Figure \ref{fig:mut}) and the maximum loads among the macro - and the small cells (Figure \ref{fig:max_loads}). The performance indicators are calculated  over the coverage area of the 3 macro cells implementing the \ac{eICIC}, which includes the 12 small cells.
	
	Figure \ref{fig:max_loads} shows that in all cases where \ac{SON} is implemented, the loads are more or less balanced between macro cells (brown bars) and small cells (white bars). Interestingly, the 'LBonly' case completely balances the loads since it is the only objective pursued in this case and the operating point where all the loads are balanced is feasible (the maximum \ac{CIO} is not limiting).
It is noted that in the 'PF1' case, the increase of \acp{CIO} is accompanied by the increase of \ac{ABSr} of the interfering macro cells, resulting in higher loads than in the 'PF2approx' case.
	
Figure \ref{fig:cet} shows that \ac{eICIC} in case 'PF1' benefits the most to users with the worst \ac{SINR}, namely provides the best \ac{CET}, although the difference with case 'PF2approx' is very small. The good \ac{CET} performance of case 'PF1' comes at the expense of a lower \ac{MUT}, as shown in  Figure \ref{fig:mut}. The highest \ac{MUT} is provided by the 'LBonly' case which better preserves macro cells' resources (see Figure \ref{fig:mut}). However this comes at the expense of fairness, i.e. lower \acp{CET} as shown in Figure \ref{fig:cet}.

The 'PF2approx' case provides the best compromise between overall network capacity and fairness, namely between \ac{MUT} and \ac{CET}. All three cases implementing \ac{SON} algorithms perform much better than the reference case with no \ac{SON} and provide better \ac{QoS}. Gain in \ac{CET} is of 50\%, 103\% and 94\% for 'LBonly', 'PF1' and 'PF2approx' solutions respectively with respect to the reference 'NoSON' solution. The gain in \ac{MUT} is of 140\%, 51\% and 101\% for 'LBonly', 'PF1' and 'PF2approx' solutions, respectively, with respect to the reference 'NoSON' solution. It is noted that only cases 3 and 4 ('PF1' and 'PF2approx') that implement each two \ac{SON} functions can be considered as \ac{eICIC}.

It is noticed that the performance gain of all algorithms presented in this paper depends on traffic density. For low traffic density, only the \ac{CET} is enhanced while at high traffic density, all the \acp{KPI} are improved by the \ac{SON} functions. Hence in practice, a threshold related to traffic demand, e.g. the cell load, can be set to decide when to trigger the SON. Finally, it is noted that all \ac{SA} based \ac{SON} algorithms developed in the paper have fast convergence time, of the order of 20 minutes for stationary traffic.
	
\begin{figure}[!ht]
\centering
\includegraphics[width=3in]{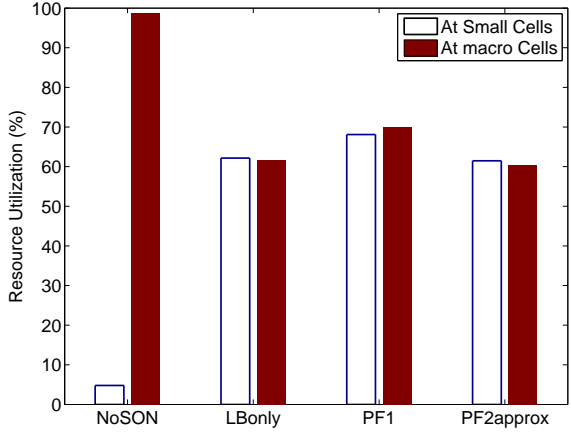}
\caption{Comparison of the maximum loads among the macro - and small cells for cases 1--4}
\label{fig:max_loads}
\end{figure}

\begin{figure}[!ht]
\centering
\includegraphics[width=3in]{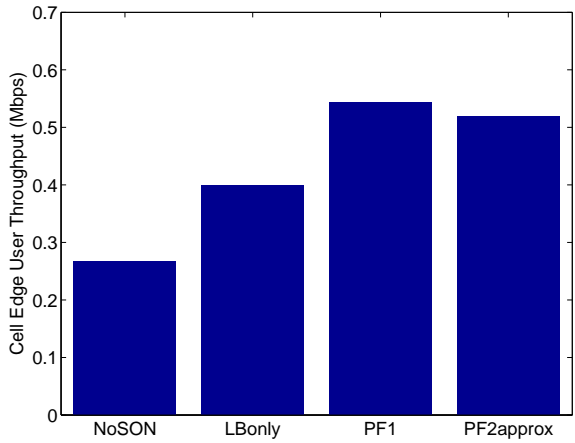}
\caption{Cell Edge Throughput improvement for cases 1--4}
\label{fig:cet}
\end{figure}

\begin{figure}[!ht]
\centering
\includegraphics[width=3in]{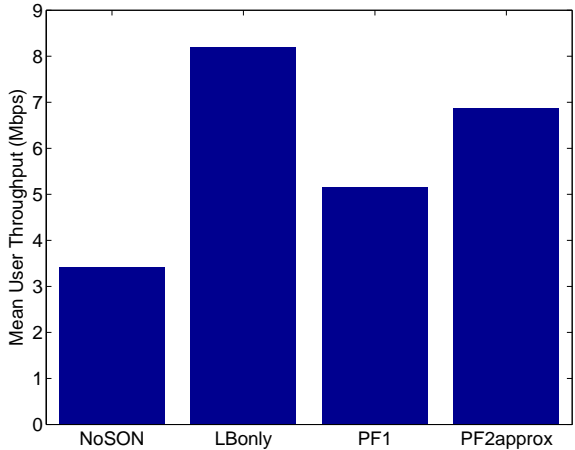}
\caption{Mean User Throughput improvement for cases 1--4}
\label{fig:mut}
\end{figure}

\section{Conclusion and Remarks} \label{sec:conclusion}

In this paper we have investigated the problem of self-optimizing \ac{eICIC} parameters, namely \ac{CRE} and \ac{ABSr} in LTE-Advanced heterogeneous networks. To this end, we have used results from \ac{SA} which provide a powerful tool for designing simple and efficient \ac{SON} algorithms. Fast convergence time is obtained for all the studied algorithms. The proposed \ac{eICIC} \ac{SON} comprizes two self-optimizing mechanisms. The first is a load balancing \ac{SON} algorithm which adapts the \ac{CRE} parameter of the small cells, and the second adapts the \ac{ABSr} to maximize a \ac{PF} utility of the user throughputs. Two possible implementations of the \ac{eICIC} mechanism have been considered. The first implementation protects users in the cell edge of the small cells (\ac{CRE} users) by scheduling them only during \acp{ABS} of the neighboring macro \acp{BS}. The second implementation allows all small cell users to benefit from the \ac{ABS} mechanism. In each case, the appropriate \ac{ABSrO} algorithm is provided. Performance results for a heterogeneous network with elastic traffic show that the second implementation of the eICIC mechanism provide better global results than the first one which only slightly outperforms in \ac{CET}. So we can say that it is globally better to let all pico \acp{UE} profit from \acp{ABS}. In practice, an opportunistic scheduler such as PF can be used along with the second implementation of \ac{eICIC}. Such implementation gives even more weight to the 'PF2approx' algorithm which is expected to provide the best performance.

\appendices
\section{Reminder on Stochastic Approximation} \label{app:rem_sa}
	Let us consider a stochastic algorithm of the type
	\begin{equation} \label{eq:generic_sa}
	x_{k+1} = x_{k} + \epsilon_k (\nabla_x f(x_k) + M_k )
	\end{equation}
	where $x \in \mathbb{R}$, $\epsilon_k \in \mathbb{R}^+_*$, $M_k$ is a noise value and $f()$ - a function we want to maximize.
	\ac{SA} (\cite[Theorem 2 P16]{Borkar}) says that if
	\begin{itemize}
	\item The series $\epsilon_k$ is non-summable; $\sum_{k=0}^\infty \epsilon_k = \infty$
	\item The series $\epsilon_k$ square summable; $\sum_{k=0}^\infty \epsilon_k^2 < \infty$
	\item $M_k$ is a Martingale difference sequence with respect to the family of $\sigma$-fields $\Sigma_k = \sigma(x_n, M_n, n \leq k )$
	\item $\sup_k \| x_k \| < \infty$ almost surely
	\end{itemize}
	then $x_k$ converges to a compact connected internally chain transitive invariant set of
	\begin{equation} \label{eq:generic_ode}
		\dot{x} = \nabla f(x)
	\end{equation}
	namely that the mean behavior of \eqref{eq:generic_sa} is described by \eqref{eq:generic_ode}.

\section{Proof of Theorem \ref{th:absr_pf1}}\label{app:absr_pf1}
	Suppose that the \ac{ABSr} $\theta$ is bounded away from 0 and 1 ($\theta \in ]0,1[$). Firstly, we can note that $U_{PF1}$ is differentiable on $]0,1[$ and let us evaluate its derivative. Using the properties of the log function ($\log(ab) = \log a + \log b$), we can rewrite $U_{PF1}(\theta)$ as
	\begin{equation} \label{eq:pf1_simple}
	\begin{split}
		U_{PF1}(\theta) = & \sum_{m=1}^M N_m \log(1 - \theta) + N_{p,\textit{CEN}} \log(1-\theta) \\
		                + & N_{p,\textit{CRE}} \log( \theta) + C_1
		\end{split}
	\end{equation}
	where
	\begin{equation*}
	\begin{split}
	C_1 = & \sum_{m=1}^M \sum_{u \in m} \log(\bar{R}_{u,m}) + \sum_{u \in \text{center of } p} \log( \bar{R}_{u,p}^{\text{no ABS}}) \\
	    + & \sum_{u \in \text{CRE of } p} \log(\bar{R}_{u,p}^{\text{ABS}})
	\end{split}
	\end{equation*}
is independent of $\theta$. From  \eqref{eq:pf1_simple}, we can easily derive
	\begin{equation}  \label{eq:pf1_der}
		\frac{\partial U_{PF1}(\theta)}{\partial \theta} = \frac{N_{p,\textit{CRE}}}{\theta} - \frac{N_{p,\textit{CEN}} + \sum_{m=1}^M N_m}{1-\theta}
	\end{equation}
Using \ac{SA} results (see Appendix \ref{app:rem_sa}), we can see that the equivalent \ac{ODE} to \eqref{eq:ABSrO_PF1_SA} is
	\begin{equation} \label{eq:ABSrO_PF1_ODE}
	\dot{\theta} = \frac{\partial U_{PF1}(\theta)}{\partial \theta}
	\end{equation}
	Since $U_{PF1}(\theta)$ is the sum of the log of concave positive functions, it is concave. So \eqref{eq:ABSrO_PF1_ODE} converges toward the maximum of $U_{PF1}(\theta)$.

\section{Proof of Theorem \ref{th:absr_pf2_exact}} \label{app:absr_pf2_exact}
The proof is similar to that of Theorem \ref{th:absr_pf1}. Except now $U_{\text{PF2\_exact}}(\theta)$ is be rewritten as
	\begin{equation} \label{eq:pf2_exact_simple}
	\begin{split}
		U_{\text{PF2\_exact}}(\theta) = & \sum_{u \in p} \log((1-\theta) \bar{R}_{u,p}^{\text{no ABS}} + \theta \bar{R}_{u,p}^{\text{ABS}}) \\
		                + & \sum_{m=1}^M N_m \log(1 - \theta) + C_2
	\end{split}	
	\end{equation}
	where $C_2 = \sum_{m=1}^M \sum_{u \in m} \log(\bar{R}_{u,m})$ is independent of $\theta$.
	
	Function \eqref{eq:pf2_exact_simple} is also a concave function of $\theta$ and its derivative reads
	\begin{equation}  \label{eq:pf2_exact_der}
		\frac{\partial U_{\text{PF2\_exact}}(\theta)}{\partial \theta} = \sum_{u \in p} \frac{1}{\theta + \frac{\bar{R}_{u,p}^{\text{no ABS}}}{\bar{R}_{u,p}^{\text{ABS}} - \bar{R}_{u,p}^{\text{no ABS}}}} - \frac{\sum_{m=1}^M N_m}{1-\theta}
	\end{equation}
	The rest of the proof follows from Appendix \ref{app:absr_pf1}.
	
\section{Proof of Theorem \ref{th:absr_pf2}} \label{app:absr_pf2}
The proof is also similar to that of Theorem \ref{th:absr_pf1}, except that $U_{PF2}(\theta)$ is written as
	\begin{equation} \label{eq:pf2_simple}
	\begin{split}
		U_{PF2}(\theta) = & \sum_{m=1}^M N_m \log(1 - \theta) + \frac{N_p}{2} \log(1-\theta) \\
		                + & \frac{N_p}{2} \log(\theta) + C_3
	\end{split}	
	\end{equation}
	where
	\begin{equation*}
	\begin{split}
	C_3 = & \sum_{m=1}^M \sum_{u \in m} \log(\bar{R}_{u,m}) + \sum_{u \in p} \frac{1}{2} \log( \bar{R}_{u,p}^{\text{no ABS}}) \\
	    + & \sum_{u \in p} \frac{1}{2} \log(\bar{R}_{u,p}^{\text{ABS}}) + N_p \log(2)
	\end{split}
	\end{equation*}
is independent of $\theta$.

	Function \eqref{eq:pf2_simple} is also concave and its derivative will then be
	\begin{equation}  \label{eq:pf2_der}
		\frac{\partial U_{PF2}(\theta)}{\partial \theta} = \frac{N_p}{2 \theta} - \frac{\frac{N_p}{2} + \sum_{m=1}^M N_m}{1-\theta}
	\end{equation}
	
	The rest of the proof follows from Appendix \ref{app:absr_pf1}.

\bibliographystyle{IEEEtran}
\bibliography{main}

\begin{thebibliography}{10}
\providecommand{\url}[1]{#1}
\csname url@samestyle\endcsname
\providecommand{\newblock}{\relax}
\providecommand{\bibinfo}[2]{#2}
\providecommand{\BIBentrySTDinterwordspacing}{\spaceskip=0pt\relax}
\providecommand{\BIBentryALTinterwordstretchfactor}{4}
\providecommand{\BIBentryALTinterwordspacing}{\spaceskip=\fontdimen2\font plus
\BIBentryALTinterwordstretchfactor\fontdimen3\font minus
  \fontdimen4\font\relax}
\providecommand{\BIBforeignlanguage}[2]{{%
\expandafter\ifx\csname l@#1\endcsname\relax
\typeout{** WARNING: IEEEtran.bst: No hyphenation pattern has been}%
\typeout{** loaded for the language `#1'. Using the pattern for}%
\typeout{** the default language instead.}%
\else
\language=\csname l@#1\endcsname
\fi
#2}}
\providecommand{\BIBdecl}{\relax}
\BIBdecl

\bibitem{3gpp.36.300-R10}
3GPP, ``{Evolved Universal Terrestrial Radio Access (E-UTRA) and Evolved
  Universal Terrestrial Radio Access (E-UTRAN); Overall description; Stage
  2},'' {3rd Generation Partnership Project (3GPP)}, TS {36.300 v10.11.0}, Sep.
  2013.

\bibitem{wang2012performance}
Y.~Wang and K.~I. Pedersen, ``{Performance analysis of enhanced inter-cell
  interference coordination in LTE-Advanced heterogeneous networks},'' in
  \emph{Vehicular Technology Conference (VTC Spring), 2012 IEEE 75th}.\hskip
  1em plus 0.5em minus 0.4em\relax IEEE, 2012, pp. 1--5.

\bibitem{pedersen2012eicic}
K.~I. Pedersen, Y.~Wang, B.~Soret, and F.~Frederiksen, ``{eICIC Functionality
  and Performance for LTE HetNet Co-Channel Deployments},'' in \emph{Vehicular
  Technology Conference (VTC Fall), 2012 IEEE}.\hskip 1em plus 0.5em minus
  0.4em\relax IEEE, 2012, pp. 1--5.

\bibitem{shirakabe2011performance}
M.~Shirakabe, A.~Morimoto, and N.~Miki, ``{Performance evaluation of inter-cell
  interference coordination and cell range expansion in heterogeneous networks
  for LTE-Advanced downlink},'' in \emph{Wireless Communication Systems
  (ISWCS), 2011 8th International Symposium on}.\hskip 1em plus 0.5em minus
  0.4em\relax IEEE, 2011, pp. 844--848.

\bibitem{pang2012optimized}
J.~Pang, J.~Wang, D.~Wang, G.~Shen, Q.~Jiang, and J.~Liu, ``Optimized
  time-domain resource partitioning for enhanced inter-cell interference
  coordination in heterogeneous networks,'' in \emph{Wireless Communications
  and Networking Conference (WCNC), 2012 IEEE}.\hskip 1em plus 0.5em minus
  0.4em\relax IEEE, 2012, pp. 1613--1617.

\bibitem{bedekar2013optimal}
A.~Bedekar and R.~Agrawal, ``{Optimal muting and load balancing for eICIC},''
  in \emph{Modeling \& Optimization in Mobile, Ad Hoc \& Wireless Networks
  (WiOpt), 2013 11th International Symposium on}.\hskip 1em plus 0.5em minus
  0.4em\relax IEEE, 2013, pp. 280--287.

\bibitem{vasudevandynamic}
S.~Vasudevan, R.~Pupala, and K.~Sivanesan, ``{Dynamic eICIC—A Proactive
  Strategy for Improving Spectral Efficiencies of Heterogeneous LTE Cellular
  Networks by Leveraging User Mobility and Traffic Dynamics},'' \emph{IEEE
  Transactions on Wireless Communications,}, vol.~12, no.~10, pp. 4956--4969,
  2013.

\bibitem{deb2013algorithms}
S.~Deb, P.~Monogioudis, J.~Miernik, and J.~P. Seymour, ``{Algorithms for
  Enhanced Inter-Cell Interference Coordination (eICIC) in LTE HetNets},''
  \emph{IEEE/ACM Transactions on Networking,}, 2013.

\bibitem{yasireICIC}
Y.~Khan, B.~Sayrac, and E.~Moulines, ``Surrogate based centralized son:
  Application to interference mitigation in lte-a hetnets,'' in \emph{IEEE 77th
  Vehicular Technology Conference (VTC Spring)}, 2013, pp. 1--5.

\bibitem{Kushner}
H.~J. Kushner and G.~G. Yin, \emph{Stochastic Approximation and Recursive
  Algorithms and Applications 2nd edition}.\hskip 1em plus 0.5em minus
  0.4em\relax Springer Stochastic Modeling and Applied Probability, 2003.

\bibitem{Borkar}
V.~S. Borkar, \emph{Stochastic Approximation: A Dynamical Systems
  Viewpoint}.\hskip 1em plus 0.5em minus 0.4em\relax Cambridge University
  Press, 2008.

\bibitem{CombesInfocom2012}
R.~Combes, Z.~Altman, and E.~Altman, ``Self-organization in wireless networks:
  a flow-level perspective,'' in \emph{Proceedings of IEEE INFOCOM}, 2012.

\bibitem{3gpp.36.300-R11}
3GPP, ``{Evolved Universal Terrestrial Radio Access (E-UTRA) and Evolved
  Universal Terrestrial Radio Access (E-UTRAN); Overall description; Stage
  2},'' {3rd Generation Partnership Project (3GPP)}, TS {36.300 v11.7.0}, Sep.
  2013.

\bibitem{jensen1906fonctions}
J.~L. W.~V. Jensen, ``Sur les fonctions convexes et les in{\'e}galit{\'e}s
  entre les valeurs moyennes,'' \emph{Acta Mathematica}, vol.~30, no.~1, pp.
  175--193, 1906.

\bibitem{3gpp.36.814-R9}
3GPP, ``{Evolved Universal Terrestrial Radio Access (E-UTRA); Further
  advancements for E-UTRA physical layer aspects},'' {3rd Generation
  Partnership Project (3GPP)}, TS {36.814 v9.0.0}, Mar. 2010.

\end{thebibliography}

\end{document}